\documentclass[%
  reprint,
  amsmath,amssymb,
  aps,
  superscriptaddress,
]{revtex4-2}

\usepackage{graphicx}
\usepackage{dcolumn}
\usepackage{bm}

\usepackage{amsthm, mathtools, amsfonts}
\usepackage{braket}
\usepackage{newtxtext,newtxmath}
\mathtoolsset{showonlyrefs}

\usepackage[table]{xcolor}
\usepackage{tcolorbox}

\usepackage[colorlinks=true,citecolor=teal,linkcolor=teal,urlcolor=teal]{hyperref}

\usepackage{algorithm}
\usepackage{algpseudocode}

\usepackage{tikz-cd}
\usetikzlibrary{backgrounds}

\usepackage{indentfirst}
\usepackage{enumitem}
\usepackage{ulem}
\usepackage{makecell}
\usepackage{booktabs}
\usepackage{multirow}
\usepackage{subcaption}
\usepackage{ragged2e}
\usepackage{cleveref}
\usepackage{graphicx}
\usepackage{csquotes}
\usepackage{nicefrac}
\usepackage{float}

\makeatletter
\long\def\@makecaption#1#2{%
 \vskip\abovecaptionskip
 {\justifying\noindent #1.\ #2\par}%
 \vskip\belowcaptionskip
}
\makeatother

\newcommand{\eqa}[1]{\begin{equation*}\begin{aligned}#1\end{aligned}\end{equation*}}

\newcommand{\BB}{\mathbb{B}}

\newcommand{\FF}{\mathbb{F}}

\newcommand{\tand}{\text{ and }}

\newcommand{\supp}{\mathrm{supp }}

\newcommand{\code}[1]{[\![#1]\!]}

\theoremstyle{definition}
\newtheorem{definition}{Definition}

\newtheorem{lemma}[definition]{Lemma}

\newtheorem{example}[definition]{Example}

\renewcommand{\emph}[1]{\textit{#1}}

\begin{document}

\preprint{APS/123-QED}

\title{Nearest-neighbour gates are all you need:\\High-rate quantum low-density parity-check codes on a planar grid}

\author{Boren Gu}
\email{b.gu@fu-berlin.de}
\affiliation{Dahlem Center for Complex Quantum Systems,
Freie Universität Berlin, Berlin, 14195, Germany}

\author{Tamas Noszko}
\email{tamas.noszko@ed.ac.uk}
\affiliation{Quantum Software Lab, The University of Edinburgh, Edinburgh, EH8 9AB, United Kingdom}

\author{Vincent Steffan}
\affiliation{IQM Quantum Computers, Munich, 80992, Germany}

\author{Jens Niklas Eberhardt}
\affiliation{Institute of Mathematics, Johannes Gutenberg-Universität Mainz, Mainz, 55128, Germany}

\author{Joschka Roffe}
\affiliation{Quantum Software Lab, The University of Edinburgh, Edinburgh, EH8 9AB, United Kingdom}

\author{Jens Eisert}
\affiliation{Dahlem Center for Complex Quantum Systems,
Freie Universität Berlin, Berlin, 14195, Germany}

\affiliation{Helmholtz-Zentrum Berlin für Materialien und Energie, Berlin, 14109, Germany}

\author{Stergios Koutsioumpas}
\email{skoutsio@ed.ac.uk}
\affiliation{Quantum Software Lab, The University of 
Edinburgh, Edinburgh, EH8 9AB, United Kingdom}

\begin{abstract}
    High-performance quantum low-density parity-check codes promise substantial reductions in the overhead of fault-tolerant quantum computation, but most constructions require long-range connectivity or qubit shuttling, both of which are difficult to realise in superconducting architectures. Here we introduce a family of quantum low-density parity-check codes that, for the first time, combines planar open-boundary layouts, finite-size advantages over surface codes, and syndrome extraction using only nearest-neighbour gates on a planar grid of qubits with either square or hexagonal connectivity. The resulting circuits achieve optimal constant-depth stabiliser measurement, independent of code size, and naturally remove leakage from the system by exchanging the roles of check and data qubits at each syndrome extraction round. We find finite-size instances such as a \(\code{323,14,15}\) code, whose code-efficiency ratio is nearly an order of magnitude larger than that of rotated surface-code patches. At around \(40\) circuit qubits per logical qubit, the best directional tile-code layouts reduce the per-logical per-round logical error rate by a factor of \(100\) relative to rotated surface-code memories. These results show that the advantages of quantum low-density parity-check codes can survive compilation into strictly planar nearest-neighbour circuits, bringing low-overhead fault-tolerant memories closer to near-term hardware.
\end{abstract}
\maketitle

\section{Introduction}

Quantum error correction is a central requirement for turning noisy quantum devices into reliable machines that can perform useful computations beyond what is feasible classically.
Indeed, fault-tolerant operations are expected to be essential for implementing algorithms beyond logarithmic depth~\cite{Roads,QECBasic,MindTheGaps,Nonunital,FG20}. The core question is therefore not whether error correction is needed, but which codes can deliver the required logical protection within the geometric and connectivity constraints of real hardware.

The surface code has long been the leading answer, especially for superconducting processors where two-qubit gates are typically restricted to nearest-neighbour interactions on a two-dimensional grid~\cite{TopologicalQuantumMemory,Kitaev-AnnPhys-2003,GoogleThreshold,Krinner_2022,Zhao_2022}. Its appeal lies in its locality: stabilisers are geometrically local and can be measured by short-depth nearest-neighbour circuits. This hardware compatibility, however, comes at the cost of a low encoding rate. A rotated surface-code patch with distance \(d\) encodes only one logical qubit into \(d^2\) data qubits, leading to substantial physical-qubit overhead for large fault-tolerant computations.

\emph{Quantum low-density parity-check}  (qLDPC) codes offer a route beyond this overhead~\cite{PRXQuantum.2.040101,Panteleev2021degeneratequantum,9996782,9567703,10.1109/TIT.2021.3097347,bravyi2024high,LiangGuChen2026,ZhouGuEisert2026}. By encoding many logical qubits in a single block while retaining bounded-weight stabilisers, qLDPC codes can in principle reduce the footprint of fault-tolerant memories far below that of separate surface-code patches. Recent constructions have made this promise increasingly concrete, with improved asymptotic resource scaling and striking finite-size reductions in physical-qubit overhead~\cite{Pinnacle,Preskill10000,zhao2026ultrahighratequantumerrorcorrection,khan2026architectingearlyfaulttolerant,wills2026concatenatingalgebraiccodeshighrate}. The central obstacle is implementation: high-performance qLDPC codes typically require check-data interactions that are non-local in any planar embedding, and hence rely on qubit shuttling or reconfigurable connectivity~\cite{dasu2026computingencodedlogicalqubits,Xu2024ConstantOverhead,LukinQEC}. In superconducting processors, one route to such connectivity is provided by flip-chip architectures, which can route signals vertically through the processor stack to couple qubits that are not nearest neighbours in the plane~\cite{bravyi2024high,Mathews_2026,Wang_2026}. This approach is in principle powerful, but it introduces additional three-dimensional integration and associated substantial fabrication complexity.

This tension has a rigorous asymptotic origin. Stabiliser codes with geometrically local checks on a two-dimensional lattice obey strong restrictions on their parameters, including bounds of the form \(kd^2=O(n)\)~\cite{BravyiPoulinTerhal,ConnectivityKrishna}. Such results rule out simultaneously high rate and large distance in the large-block-length limit under static two-dimensional locality. They do not, however, settle the finite-size question most relevant to early fault-tolerant devices~\cite{low2026denserplanarsurfacecode}: can one obtain a substantial overhead advantage over the surface code while retaining a strictly planar nearest-neighbour implementation?

\begin{figure*}
    \centering
    \includegraphics[width=\linewidth]{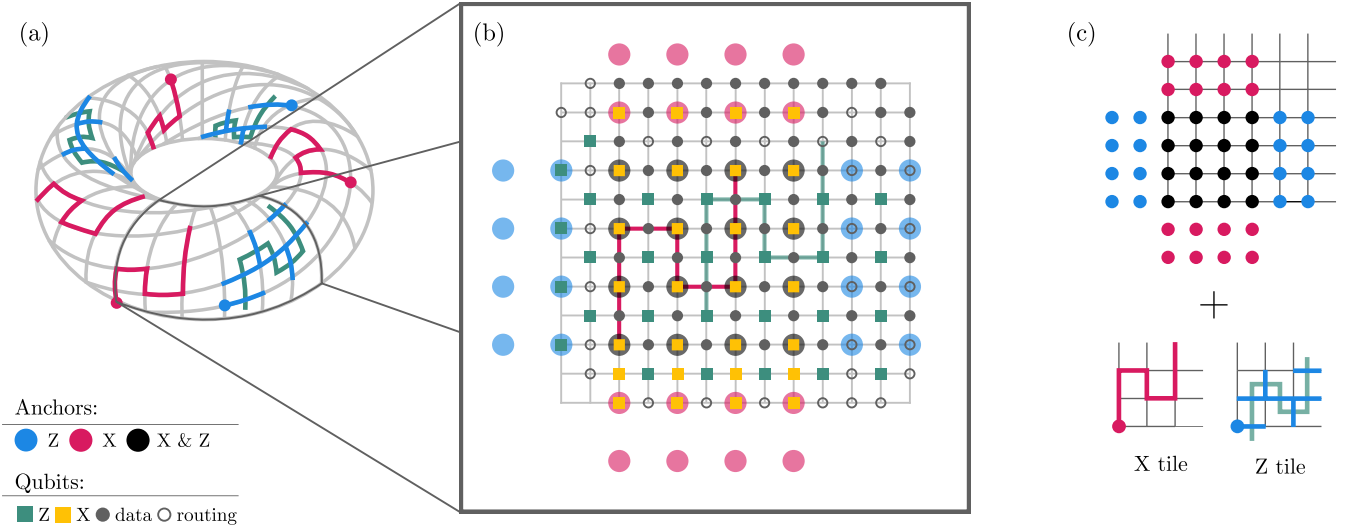}
    \caption{\textbf{Planar directional tile codes.}
    A directional tile code can be viewed either as an open-boundary patch cut from a toric directional code~\cite{GeherByfieldRuban2025}, or as a planar tessellation by connected-string tiles in the spirit of tile codes~\cite{steffan2025high,qv65-vmzr}.
    \textbf{(a) Toric directional code.}
    A toric directional code on a twisted torus, with data qubits on edges and \(X\)- and \(Z\)-type checks on vertices and plaquettes. Each stabiliser follows the directional word \(\mathfrak D=\texttt{N}^2\texttt{ESE}\texttt{N}^2\).
    \textbf{(b) Planar directional tile code (this work).}
    The corresponding open-boundary \(\texttt{N}^2\texttt{ESE}\texttt{N}^2\) directional tile code with parameters \(\code{60,4,5}\), embedded in a hardware-like square-grid layout. Dark grey dots are data qubits, coloured squares are check qubits, and unfilled circles are routing sites for nearest-neighbour measurement walks. Highlighted \(X\)- and \(Z\)-type stabilisers have the same connected-string shape and are measured by simultaneous CXSWAP walks. The displayed layout is routing-optimised, with shifted and pruned boundary qubits; see~\Cref{fig:routing_optimisation}. Background colours mark tile anchors.
    \textbf{(c) Tile-code construction.}
    Equivalent planar tessellation by directional \(X\)- and \(Z\)-tiles whose supports form the same connected string on primal and dual lattices.}
    \label{fig:1}
\end{figure*}

We answer this question affirmatively by introducing \emph{directional tile codes}: planar qLDPC memories whose syndrome-extraction circuits are implemented on a square or hexagonal grid layout using only nearest-neighbour \(\mathrm{iSWAP}\) gates~\cite{eickbusch_demonstration_2025,Sung_2021,Picard_2024,chen2025efficient,inoue2026parametricallydriveniswapgate,tiwari2026highfidelityiswapgatedouble,McEwen2023relaxinghardware}, up to single-qubit Clifford rotations. The method leverages the \(\mathrm{iSWAP}\) interaction as a spacetime routing primitive, extending ideas from dynamic surface codes~\cite{McEwen2023relaxinghardware,eickbusch_demonstration_2025} and toric directional codes~\cite{GeherByfieldRuban2025}. During syndrome extraction, the \(\mathrm{iSWAP}\)-induced SWAP dynamics moves quantum information through the plane while simultaneously generating the required check-data interactions. Each stabiliser is therefore measured as a nearest-neighbour walk on the hardware graph, rather than through a static long-range coupling graph. 

Due to this mechanism, directional tile codes are especially promising with regard to leakage noise~\cite{Fowler_2013,Miao_2023}, a critical bottleneck for running quantum error correction protocols on superconducting hardware. On the one hand, \(\mathrm{iSWAP}\) gates can naturally minimise leakage generation and transport~\cite{eickbusch_demonstration_2025}. On the other hand, our dynamic syndrome extraction circuits for directional tile codes exchange data and check qubits after each round of syndrome extraction, similarly to walking circuits for surface codes~\cite{McEwen2023relaxinghardware}. In this way, both data and check qubits are naturally reset every other round, which has been demonstrated to effectively remove the population from higher excited states~\cite{eickbusch_demonstration_2025}.

The relevant benchmark is twofold. At the code level, one would like finite-size parameters that improve substantially over rotated surface-code patches~\cite{Bombin_2007}; a convenient measure is the \textit{code-efficiency ratio}, \(kd^2/n\), which equals \(1\) for a rotated surface-code patch with parameters \(\code{d^2,1,d}\). At the implementation level, this advantage must survive after compiling syndrome extraction into a local circuit, including all data, check, and potential routing qubits required by the hardware layout.

Directional tile codes perform strongly under both benchmarks. On the square grid, our search identifies compact instances such as the \(\code{323,14,15}\) code, whose code-efficiency ratio is nearly an order of magnitude larger than that of a rotated surface-code patch at the same distance; see~\Cref{tab:code_para}. The construction further extends to a hexagonal layout with only degree-$3$ connectivity (Fig.~\ref{fig:hex}), yielding, for example, a \(\code{199,10,7}\) code. At the implemented-circuit level, after accounting for data, check, and routing qubits, circuit-level simulations show that the resulting nearest-neighbour layouts can outperform rotated surface-code memories at a realistic physical error rate \(p=0.001\). Across square- and hexagonal-grid connectivity, the best directional tile-code layouts both achieve up to two orders of magnitude lower per-logical per-round logical error rates than square-grid rotated surface-code memories at comparable footprints; see~\Cref{fig:footprint}.

\section{Planar qLDPC memories from\\
directional words}

\begin{figure}[]
    \centering
    \includegraphics[width=\linewidth]{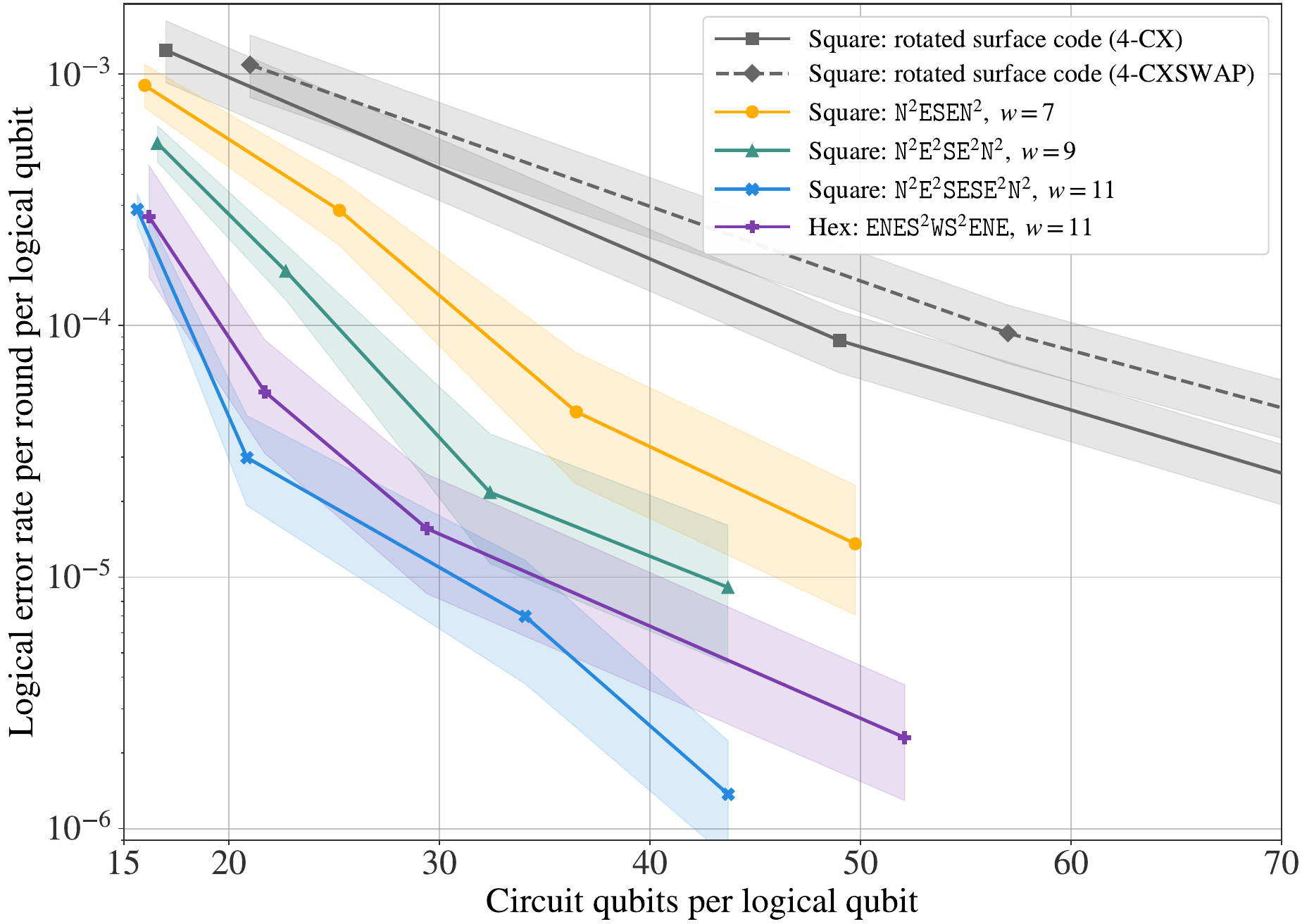}
    \caption{\textbf{Implemented footprint advantage over rotated surface-code memories on square and hexagonal grids.}
    Circuit-level logical error rate averaged over $X$ and $Z$ bases per syndrome-extraction round per logical qubit at physical error rate \(p=0.001\), obtained from simulations of \(d\) consecutive syndrome-extraction rounds. Directional tile codes on the square grid with \(w=7,9,11\) and on the hexagonal grid with \(w=11\) are compared with two square-grid rotated surface-code implementations, based on \(\mathrm{CX}\) and \(\mathrm{CXSWAP}\) gates, respectively. For comparison at a fixed encoded memory size of \(140\) logical qubits, these families require \(140\) rotated surface-code patches; \(35\), \(14\), and \(10\) square-grid directional tile-code patches for \(w=7,9,11\), respectively; or \(14\) hexagonal-grid directional tile-code patches. The choice \(140\) is the least common multiple of the logical dimensions of the codes considered. The horizontal axis gives the number of circuit qubits per logical qubit, including data, check, and routing qubits in the implemented nearest-neighbour syndrome-extraction circuits. Shaded regions indicate statistical uncertainty. Representative directional tile-code instances are shown in~\Cref{tab:code_para}, and the full set of simulated codes is listed in~\Cref{tab:full_sims_codes}.}
\label{fig:footprint}
\end{figure}

The construction is organised around \emph{directional words}. A directional word is an ordered string of lattice steps. Geometrically, it traces a bounded-size connected string on the lattice; algebraically, this string defines the support of a stabiliser; operationally, the same ordering specifies the local measurement schedule. The paired \(X\)- and \(Z\)-type checks are obtained by placing the same string on primal and dual lattices, respectively, so that one directional word defines both the \textit{Calderbank–Shor–Steane} (CSS) stabiliser structure and the spacetime walk used to measure it. Figure~\ref{fig:1} illustrates this principle for the word \(\mathfrak D=\texttt{N}^2\texttt{ESE}\texttt{N}^2\), while \Cref{fig:directions} shows the nearest-neighbour physical-action layers associated with the four lattice directions.

\begin{table}[]
{
\setlength{\aboverulesep}{0pt}
\setlength{\belowrulesep}{0pt}
\renewcommand{\arraystretch}{1.2}
\begin{tabular}{l|c|c|c|c}
\toprule
Code family& NN & 2D planar & $\eta_{\mathrm{circ}} >5$&Two-qubit gates\\
\midrule
\midrule
Strictly local tile~\cite{BaePark2026LocalTileCodes}& $\checkmark$ & $\checkmark$ & $\times$&$\mathrm{CP,SWAP}$\\
Bunny~\cite{Zhou2026BunnyCodes}& $\checkmark$ & $\checkmark$ & $?^\ddagger$&$\mathrm{CP,iSWAP}$\\
Vine~\cite{Nixon2026VineCodes}& $\checkmark$ & $\checkmark$ & $\checkmark$&$\mathrm{CP,SWAP,iSWAP}$\\
Barbell~\cite{choe2026barbellcodesqldpccodes} & $\times$ & $\checkmark$ & $\checkmark$ &$\mathrm{CP}$\\
Denser surface~\cite{low2026denserplanarsurfacecode} & $\checkmark$ & $\checkmark$ &$\times^\dagger$&$\mathrm{CP}$\\
Tile~\cite{steffan2025high,qv65-vmzr} & $\times$ & $\checkmark$ & $\checkmark$&$\mathrm{CP}$\\
Toric directional~\cite{GeherByfieldRuban2025} & $\checkmark$ & $\times$ & $?^\ddagger$&$\mathrm{iSWAP}$\\
\textbf{This work} & $\checkmark$ & $\checkmark$ & $\checkmark$&$\mathrm{iSWAP}$\\

\bottomrule
\end{tabular}
\caption{\textbf{Comparison of code families by nearest-neighbour (NN) implementability, planar open-boundary geometry, circuit-efficiency ratio and required two-qubit gates.}
The circuit-efficiency ratio \(\eta_{\mathrm{circ}}=k\,n_{\mathrm{RSC}(d)}/n_{\mathrm{circ}}\) benchmarks the implemented layout against \(k\) rotated surface-code patches with distance $d$. Here \(n_{\mathrm{circ}}\) includes all physical qubits used by the syndrome-extraction layout.
\(^{\dagger}\)For the denser surface-code construction, the non-yoked dense-packing architecture has \(\eta_{\mathrm{circ}}\lesssim 2\). The larger quoted densities, up to \(\eta_{\mathrm{circ}}\lesssim 4.5\), rely on yoking, which introduces additional parity-check measurements, workspace, latency, and idealised outer-decoder assumptions~\cite{low2026denserplanarsurfacecode}.
\(^{\ddagger}\)Refs.~\cite{rowshan2026structuralanalysisdirectionalqldpc} and~\cite{Zhou2026BunnyCodes} searched only up to stabiliser weights \(7\) and \(6\), respectively. We show that allowing for slightly larger stabiliser weights substantially improves this ratio.}
}
\label{tab:code_family_comparison}
\end{table}

\begin{table}[t]
\centering
{
\setlength{\aboverulesep}{0pt}
\setlength{\belowrulesep}{0pt}
\renewcommand{\arraystretch}{1.35}

\begin{tabular}{c|c|c|c|c}
\toprule
\makecell{$\mathfrak D,\ w$}
& Grid
& $\code{n,k,d}$
& $kd^2/n$
& LER$(p=0.001)$
\\
\midrule
\midrule

\multirow{2}{*}[-0.3ex]
{\makecell{$\texttt{N}^2\texttt{ESE}\texttt{N}^2,\ 7$}}
&
\multirow{6}{*}[-0.3ex]{sq.}
&
$\code{60,4,5}$
&
1.67
&
$4.54\times10^{-5}$
\\

\cmidrule{3-5}

&
&
$\code{180,4,9}$
&
1.80
&
$-$
\\

\cmidrule{1-1}
\cmidrule{3-5}

\multirow{2}{*}[-0.3ex]
{\makecell{$\texttt{N}^2\texttt{E}^2\texttt{S}\texttt{E}^2\texttt{N}^2,\ 9$}}
&
&
$\code{162,10,6}$
&
2.22
&
$9.10\times10^{-6}$
\\

\cmidrule{3-5}

&
&
$\code{351,10,9}$
&
2.31
&
$-$
\\

\cmidrule{1-1}
\cmidrule{3-5}

\multirow{2}{*}[-0.3ex]
{\makecell{$\texttt{N}^2\texttt{E}^2\texttt{SES}\texttt{E}^2\texttt{N}^2,\ 11$}}
&
&
$\code{241,14,12}$
&
8.37
&
$1.37\times10^{-6}$
\\

\cmidrule{3-5}

&
&
$\code{323,14,15}$
&
9.75
&
$-$
\\

\midrule
\midrule

\multirow{2}{*}[-0.3ex]
{\makecell{$\texttt{ENE}\texttt{S}^2\texttt{W}\texttt{S}^2\texttt{ENE},\ 11$}}
&
\multirow{2}{*}[-0.3ex]{hex.}
&
$\code{52,10,3}$
&
1.73
&
$2.71\times10^{-4}$
\\

\cmidrule{3-5}

&
&
$\code{199,10,7}$
&
2.46
&
$2.30\times10^{-6}$
\\

\bottomrule
\end{tabular}
}
\caption{\textbf{Small examples of directional tile codes and their circuit-level performance.}
For each directional word $\mathfrak D$, $w$ denotes the stabiliser weight. ``sq.'' and ``hex.'' denote square- and hexagonal-grid connectivity of degree $4$ and $3$, respectively. All listed codes admit a nearest-neighbour syndrome-extraction circuit of depth $w$, or $w+2$ including check- and routing-qubit initialisation and measurement. LER denotes the logical error rate per logical qubit per syndrome-extraction round at $p=0.001$. Monte Carlo simulations were performed only for codes with at most $60$ circuit qubits per logical qubit; entries marked by ``$-$'' were not simulated.}
\vspace{-1em}
\label{tab:code_para}
\end{table}

\begin{figure}
    \centering
    \includegraphics[width=\linewidth]{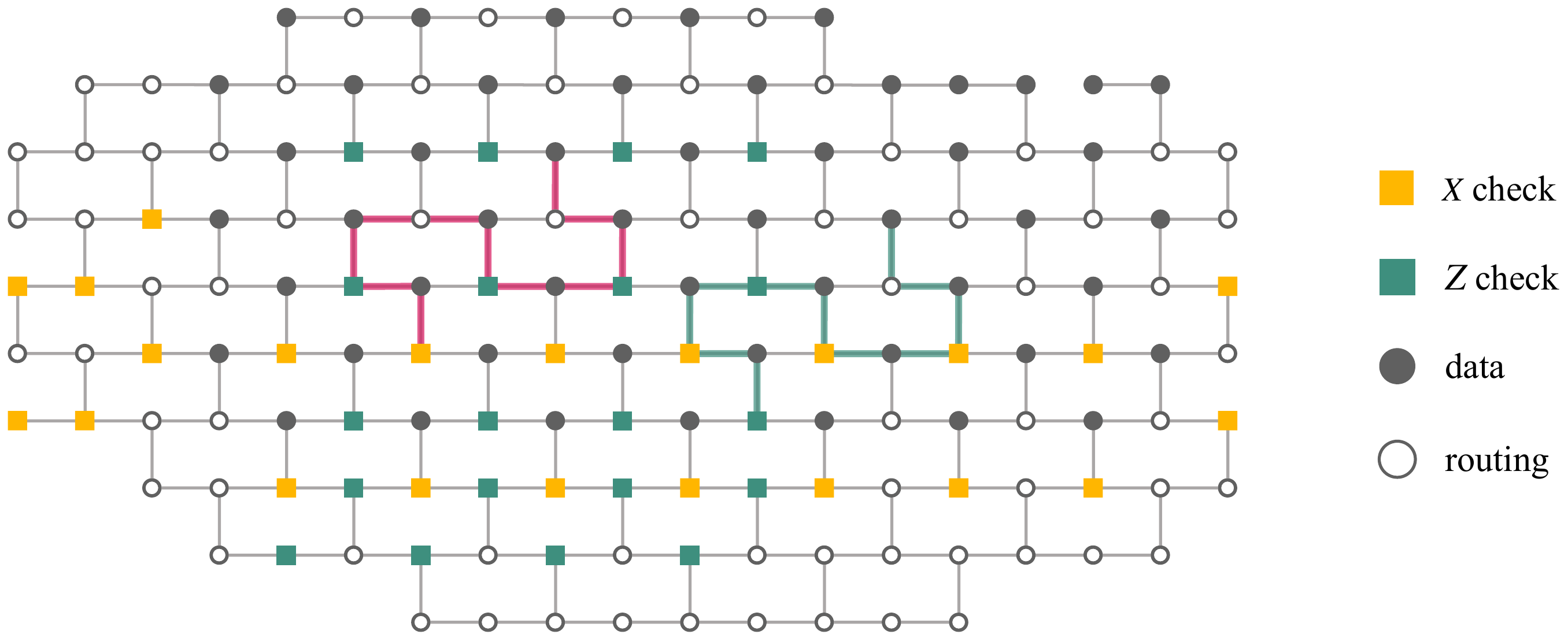}
    \caption{\textbf{Directional tile codes on the hexagonal grid.} An $\texttt{ENE}\texttt{S}^2\texttt{W}\texttt{S}^2\texttt{ENE}$ directional tile code with parameters $\code{52,10,3}$. The layout is rotated by $90^\circ$ for presentation.}
    \label{fig:hex}
    \vspace{-1em}
\end{figure}

We obtain open-boundary planar memories by using directional words as local tessellation rules. Earlier directional-code constructions realised these rules as translationally invariant stabiliser patterns, with finite codes obtained by imposing periodic boundary conditions~\cite{GeherByfieldRuban2025,rowshan2026structuralanalysisdirectionalqldpc}. Our approach is different: we apply the same local directional prescription directly to a finite planar patch, yielding hardware layouts with open boundaries rather than periodic ones. The key step is not merely to place paired connected-string \(X\)- and \(Z\)-tiles on an open patch, in the spirit of \textit{tile codes}~\cite{steffan2025high,qv65-vmzr,Liang2026Boundaries}, but to do so while preserving the ordered walk specified by the directional word. This ensures that every bulk and boundary stabiliser remains measurable by the same nearest-neighbour \(\mathrm{iSWAP}\)-based dynamics. In the bulk, data qubits occupy lattice edges and \(X\)- and \(Z\)-type check qubits occupy vertices and plaquettes, respectively. Near the boundary, check and data positions may be shifted along the directional words from the standard locations as part of the routing optimisation. Routing qubits are inserted only where required to keep the ordered walk nearest-neighbour on the physical square or hexagonal grid. For a directional word of weight \(w\), one syndrome-extraction round has depth \(w+2\), including check-qubit initialisation and measurement, independent of the overall code size. Thus the same bounded object gives both bounded-weight stabilisers and optimal-depth nearest-neighbour extraction circuits.

Our construction shows that the finite-size advantages of qLDPC codes can be retained under the same square-grid nearest-neighbour constraints that underpin the hardware compatibility of surface codes, while a hexagonal-grid architecture can further relax the required connectivity. After including the routing overhead required for nearest-neighbour syndrome extraction, the implemented memories can still outperform rotated surface-code footprint curves under a uniform circuit-level noise model. This demonstrates that qLDPC overhead reductions need not rely on long-range couplers, qubit shuttling, or non-planar connectivity: the required check-data interactions can instead be generated dynamically in spacetime by native \(\mathrm{iSWAP}\)-based walks over an open-boundary plane. Directional tile codes therefore provide a route toward low-overhead fault-tolerant memories that preserve strict two-dimensional nearest-neighbour locality while moving beyond the low encoding rate of the surface code.

\section{Nearest-neighbour syndrome extraction}
\label{sec:syndrome_extraction}

\begin{figure}
    \centering
    \includegraphics[width=\linewidth]{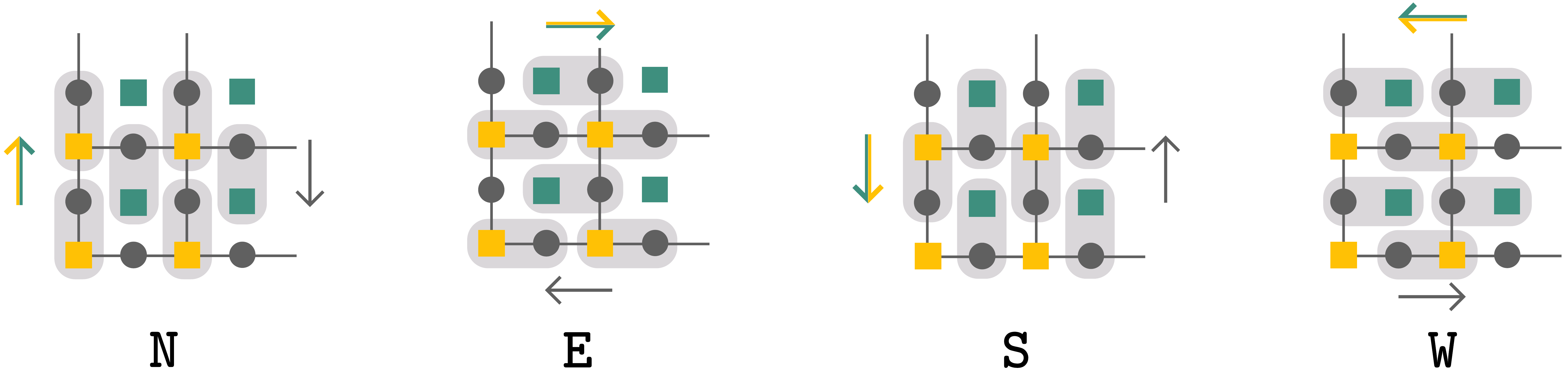}
    \caption{\textbf{Nearest-neighbour directional layers~\cite{GeherByfieldRuban2025}.} Local physical action associated with the four elementary directions \(\{\texttt{N},\texttt{E},\texttt{S},\texttt{W}\}\) in a directional word. Coloured check qubits move along the indicated direction by nearest-neighbour exchange operations, while interacting with data qubits via $\mathrm{CX}$ gates whenever the walk encounters a data qubit. Specifically, we apply $\mathrm{CXSWAP}(X\text{-check},\text{data})$ and $\mathrm{CXSWAP}(\text{data},Z\text{-check})$ respectively for different types of check-data pairs.
    } \label{fig:directions}
\end{figure}

We now make the \(\mathrm{iSWAP}\)-native syndrome-extraction circuit explicit. The input is a \textit{directional word}
\[
    \mathfrak D=\vec d_1\vec d_2\cdots \vec d_w,\qquad
    \vec d_i\in\{\texttt{N},\texttt{E},\texttt{S},\texttt{W}\},
\]
together with a square- or hexagonal-grid hardware layout consisting of data qubits \(D\), \(X\)- and \(Z\)-type check qubits \(C_X,C_Z\), and routing qubits \(R\). The word \(\mathfrak D\) specifies a sequence of nearest-neighbour directions. During syndrome extraction, each check qubit follows this sequence as a spacetime walk: whenever the walk encounters a data qubit, the circuit accumulates the corresponding stabiliser parity, and whenever it encounters a routing site, the quantum state is moved to keep the walk nearest-neighbour. Thus the ordered support of the stabiliser also guides the local schedule used to measure it.

The circuit uses only nearest-neighbour \(\mathrm{iSWAP}\) gates and single-qubit Clifford rotations. We express each data-check interaction in terms of a \(\mathrm{CXSWAP}\) gate, locally equivalent to an \(\mathrm{iSWAP}\) gate~\cite{yoshida2025lowdepthcolorcode,McEwen2023relaxinghardware}
\eqa{
\mathrm{CXSWAP}(q_0,q_1)
&=(S^\dagger\otimes HS^\dagger)\,
  \mathrm{iSWAP}(q_0,q_1)\,
  (H\otimes I),\\
&=\mathrm{SWAP}(q_0,q_1)\,\mathrm{CX}(q_0,q_1).
}
The $\mathrm{CX}$ part accumulates the stabiliser eigenvalue on the check qubit, while the $\mathrm{SWAP}$ part advances the relative positions of the check, data and routing qubits. For \(X\)-type checks the check qubit acts as the control, whereas for \(Z\)-type checks the orientation is reversed, after the usual basis change; see \Cref{fig:directions}.

One round of syndrome extraction is given in~\Cref{alg:syndrome_extraction}. The algorithm iterates through the directional word \(\mathfrak D=\vec d_1\cdots \vec d_w\), applying one nearest-neighbour directional layer for each step \(\vec d_i\). Consecutive rounds alternate between the word \(\mathfrak D\) and the inverse word, so that the physical layout is restored without introducing long-range operations.

\begin{algorithm}[H]
\caption{\textbf{Nearest-neighbour syndrome extraction}}
\label{alg:syndrome_extraction}
\begin{algorithmic}[1]
\Require Directional word \(\mathfrak D=\vec d_1\cdots \vec d_w\) and a square or hexagonal grid layout with data qubits \(D\), check qubits \(C_X,C_Z\), and routing qubits \(R\)
\Ensure One syndrome-extraction round using only nearest-neighbour \(\mathrm{iSWAP}\)-based operations

\State Initialise \(C_X\) in \(\ket{+}\) and \(C_Z\) in \(\ket{0}\)

\For{$i=1,\dots,w$}
    \ForAll{$q\in C_X\cup C_Z$}
        \If{$q^{\vec d_i}\in D$} \Comment{neighbour of \(q\) in direction \(\vec d_i\)}
            \If{$q\in C_X$}
                \State apply \(\mathrm{CXSWAP}(q,q^{\vec d_i})\)
            \Else
                \State apply \(\mathrm{CXSWAP}(q^{\vec d_i},q)\)
            \EndIf
        \ElsIf{$q^{\vec d_i}\in R$}
            \State apply \(\mathrm{SWAP}(q,q^{\vec d_i})\)
        \EndIf
    \EndFor

    \ForAll{$q\in R$}
        \If{$q^{\vec d_i}\in D$}
            \State apply \(\mathrm{SWAP}(q,q^{\vec d_i})\)
        \EndIf
    \EndFor
\EndFor

\State Measure \(C_X\) in the \(X\) basis and \(C_Z\) in the \(Z\) basis
\end{algorithmic}
\end{algorithm}

\begin{figure*}
    \centering
    \includegraphics[width=\linewidth]{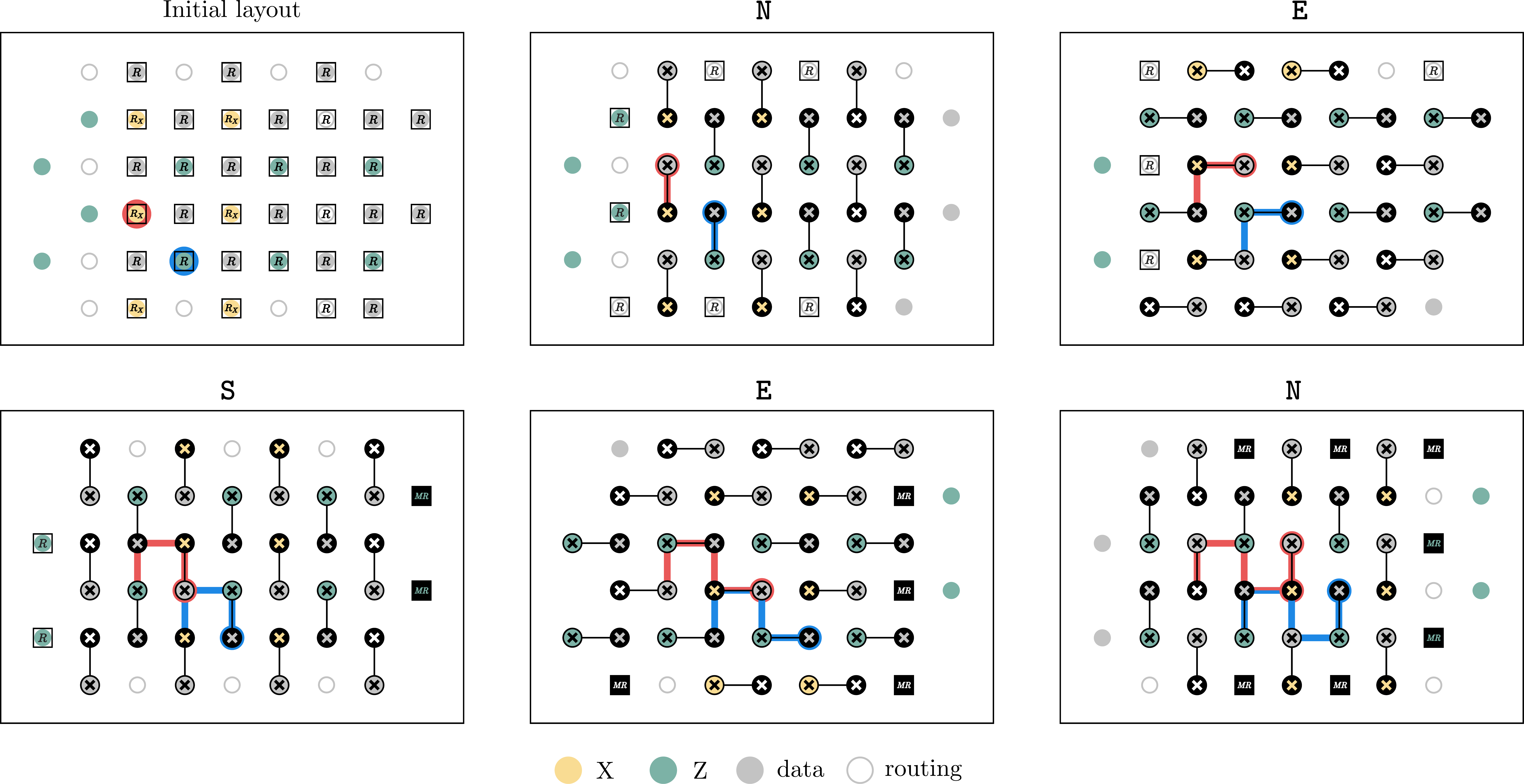}
    \caption{\textbf{Example of one round of nearest-neighbour syndrome extraction.}
    Spacetime trajectory of the nearest-neighbour CXSWAP syndrome-extraction circuit for a $\code{18,2,dx =3, dz =2}$ directional tile code on a square grid with directional word
    \(\mathfrak D=\texttt{NESEN}\). The top-left panel shows the initial layout, while the following panels show the five directional layers \(\texttt{N}\), \(\texttt{E}\), \(\texttt{S}\), \(\texttt{E}\), and \(\texttt{N}\). In each layer, check qubits are moved along the prescribed direction by parallel nearest-neighbour CXSWAP operations, as described in Algorithm~\ref{alg:syndrome_extraction}. The red and blue paths highlight representative \(X\)- and \(Z\)-type check trajectories during the same extraction round. The circuit shown includes boundary scheduling optimisations: qubits whose first interaction occurs only in a later layer are reset only immediately before use, and check qubits whose final interaction has already occurred are measured as soon as possible. Consequently, some reset and measurement operations appear inside the directional sequence rather than only at the beginning or end of the round; these scheduling choices reduce idle time without changing the stabilisers being measured.}
    \label{fig:nesen_syndrome_extraction}
\end{figure*}

The $\mathrm{SWAP}$ gates in~\Cref{alg:syndrome_extraction} are also implemented using the same $\mathrm{iSWAP}$ primitive, avoiding a bare $\mathrm{SWAP}$, which is less natural on many superconducting platforms~\cite{chen2025efficient,PRXQuantum.5.020338,Kri_an_2025}.
For each routing move, the routing qubit is prepared in $\ket{0}$ or $\ket{+}$ according to whether it is the control or target of the  $\mathrm{CXSWAP}$.
In either case, this component acts trivially, so the ideal operation reduces to a $\mathrm{SWAP}$:
\begin{align}
        \mathrm{CXSWAP}(\ket{0}_r,q)&=\mathrm{SWAP}(\ket{0}_r,q),\\
        \mathrm{CXSWAP}(q,\ket{+}_r) &=\mathrm{SWAP}(q,\ket{+}_r),
\end{align}
where \(\ket{0}_r\) or \(\ket{+}_r\) denotes the routing qubit and \(q\) is a data or check qubit. As such, both parity-accumulating interactions and routing moves are compiled into the same nearest-neighbour \(\mathrm{iSWAP}\)-based gate family.

Faults on routing qubits, or faults occurring during routing steps, can propagate to data or check qubits. This propagation remains bounded: in one syndrome-extraction round, each data or check qubit participates in at most \(w\) two-qubit interactions, including those involving routing qubits; see Lemma~\ref{lem:check_weight} in Appendix~\ref{app:algebraic_rep}. Routing qubits also provide additional control points in the circuit. They may be reset after each routing step, or measured so that non-trivial outcomes supply flag information for faults occurring along the routing path; see Appendix~\ref{app:simulation}. For a directional word of weight \(w\), one syndrome-extraction round has optimal depth \(w\), or depth \(w+2\) including check-qubit initialisation and measurement.

This exchange-based formulation is also well aligned with recent leakage-aware approaches to superconducting quantum error correction. Time-dynamic surface-code circuits have shown that swapping data and measurement roles can help remove accumulated non-computational population, and related experiments have demonstrated surface-code cycles using \(\mathrm{iSWAP}\)-family gates~\cite{McEwen2023relaxinghardware,eickbusch_demonstration_2025}. Our circuits naturally support the same data-check exchange mechanism and, moreover, allow routing qubits to be reset or measured during the walk, providing built-in locations where leakage-mitigation or flagging strategies could be incorporated. We do not model leakage below, so we use this point only to highlight compatibility with existing leakage-mitigation strategies.


\section{Directional tile codes}

A \textit{tile} is a subset of edges of a \(B\times B\) grid, excluding the edges on the top row and the rightmost column. We identify edges with physical data qubits, so that a tile specifies the support of an \(X\)-type or \(Z\)-type stabiliser check. To ensure that the corresponding \(X\)- and \(Z\)-type stabilisers commute, the two tiles are required to satisfy the \textit{mutual condition}: for each horizontal, respectively vertical, edge of the \(X\)-tile with coordinate \((a,b)\), the \(Z\)-tile contains the vertical, respectively horizontal, edge with coordinate \((B-1-a,B-1-b)\)~\cite{steffan2025high}.

A pair of \(X\)- and \(Z\)-tiles is called \textit{directional} if the support of one tile forms an ordered connected string, and the other tile contains the same string on the dual lattice. Such a pair is labelled by the corresponding sequence of lattice directions, the \textit{directional word}, \(\mathfrak D=\vec d_1\vec d_2\cdots \vec d_w\).
For example, the \(\texttt{N}^2\texttt{ESE}\texttt{N}^2\)-tiles are shown in \Cref{fig:1}~(c).

\begin{figure}[]
  \centering
  \includegraphics[width=\linewidth]{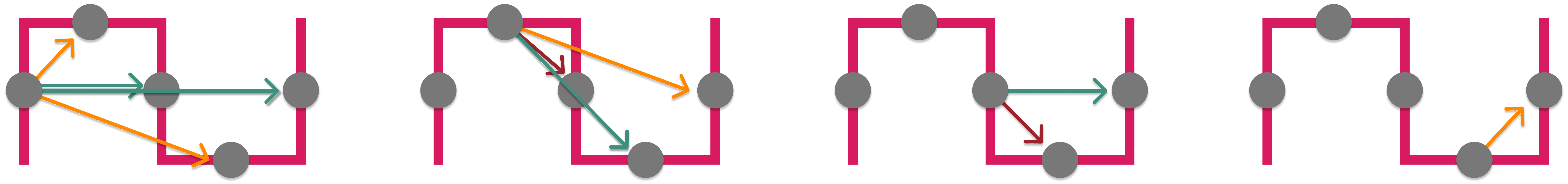}
  \caption{\textbf{Examples of displacement vectors}. The connected string of the \(X\)-tile of the directional word $\mathfrak D=\texttt{NESEN}$, with data qubits highlighted by grey dots. Displacement vectors are illustrated for each starting point. Green vectors have even vertical displacement, while the different shades of yellow indicate pairs of vectors with odd vertical displacement.}
  \label{fig:displacement_vec}
\end{figure}

The mutual condition guarantees commutation of the static stabiliser generators, but by itself it does not guarantee that the syndrome-extraction circuit in \Cref{alg:syndrome_extraction} is deterministic. We therefore impose an additional parity condition, obtained as a simplified form of Theorem~1 with Layout~1 of Ref.~\cite{GeherByfieldRuban2025}: every displacement vector of the ordered connected string with odd vertical displacement must occur an even number of times. Here, a displacement vector is the vector from one edge in the ordered support of the connected string to a later edge, with the ordering fixed by the directional word; see \Cref{fig:displacement_vec}.

\begin{definition}[Directional tiles]
\label{def:directional_tiles}
A pair of \(X\)- and \(Z\)-tiles is called \emph{directional} if the tiles satisfy the mutual condition, one tile forms an ordered connected string labelled by a directional word \(\mathfrak D=\vec d_1\vec d_2\cdots \vec d_w\), and the other tile contains the same string on the dual lattice. In addition, every displacement vector of the ordered connected string with odd vertical displacement must occur with even multiplicity.
\end{definition}

The \textit{directional tile codes} are obtained by tessellating a pair of directional tiles on an \((M+B-1)\times (N+B-1)\) rectangular grid, again with the edges on the top row and the rightmost column removed. Starting from the bottom-left vertex, we choose the vertices of an \(M\times N\) subgrid as \textit{anchors}. In the bulk, both the \(X\)- and \(Z\)-tiles are supported at each anchor. Along the boundary, \(X\)- and \(Z\)-type tiles are placed separately, as shown in \Cref{fig:1} (b) and (c), so that any \(X\)-boundary tile and \(Z\)-boundary tile either have disjoint support or intersect only on edges contained in the prescribed layout.

After tessellation, we prune unnecessary boundaries. Edges that are contained in tiles of at most one type are removed, together with any stabiliser tiles whose support becomes empty after this pruning step. The remaining edges define the data qubits, while the remaining tile anchors define the positions of the \(X\)- and \(Z\)-type check qubits. This gives a complete CSS code.

To realise the resulting directional tile code on a planar grid device using only nearest-neighbour iSWAP gates, we embed the data and check qubits into the square or hexagonal-grid hardware layout. Near the boundary, the tessellation can leave gaps between check qubits and the data qubits with which they must interact during the ordered walk. These gaps are filled with \textit{routing qubits}, shown as uncoloured circles in \Cref{fig:1}~(b) and \Cref{fig:routing_optimisation}, so that every step of the syndrome-extraction circuit remains nearest-neighbour.

We observe that logical operators in directional tile codes have a surface-code-like boundary-to-boundary structure; see~\Cref{fig:canon_logicals}. Consistently with this picture, the \(X\)- and \(Z\)-type distances scale linearly with the two bulk dimensions and are primarily controlled by \(M\) and \(N\), respectively; see~\Cref{tab:routing_dis_M_N}. For uniform-distance instances, \(d=d_X=d_Z\), this gives a natural way to compare the implemented circuit size \(n_{\mathrm{circ}}\) as a function of \(d\); see~\Cref{fig:uniform_routing_vs_distance}. For the lower-weight directional words \(\texttt{N}^2\texttt{ESE}\texttt{N}^2\) and \(\texttt{N}^2\texttt{E}^2\texttt{S}\texttt{E}^2\texttt{N}^2\), of weights \(7\) and \(9\), respectively, the resulting circuit-qubit count follows an exact quadratic dependence on \(d\), as expected from a two-dimensional patch whose two linear dimensions both grow with distance. By contrast, for the weight-$11$ word \(\texttt{N}^2\texttt{E}^2\texttt{SES}\texttt{E}^2\texttt{N}^2\), the same finite-size data are much better described by an almost linear dependence of \(n_{\mathrm{circ}}\) on \(d\). This finite-size behaviour remains consistent with known two-dimensional locality bounds~\cite{DelfosseBeverlandTremblay2021}, while revealing a favourable finite-distance regime relevant to near-term overhead comparisons.

Beyond quantum memory, the tile-code structure provides routes to logical Clifford operations. Using the canonical logical Pauli basis of tile codes~\cite{breuckmann2025logicaloperatorsderivedautomorphisms}, we analyse operations arising from derived automorphisms and geometric code symmetries. Derived automorphisms temporarily enlarge the patch along one boundary and induce logical Clifford actions that decompose into products of logical $\mathrm{CNOT}$ gates. For suitable palindromic directional words, reflection symmetries also define code automorphisms whose logical actions decompose into products of $\mathrm{SWAP}$ and $\mathrm{CNOT}$ gates~\cite{sayginel2025automorphisms}. These mechanisms offer candidates for intra-patch logical operations, complementing inter-patch lattice surgery~\cite{yang2025planarfaulttolerantquantumcomputation,HorsmanSurgery,QGPU}. Appendix~\ref{app:automorphism} describes their logical action and additional symmetries identified through automated searches. Assessing their full fault-tolerant implementation cost under planar nearest-neighbour constraints remains an important direction for future work.

\section{Routing overhead and optimisation}
The routing overhead introduced by our scheme is boundary dominated and scales linearly with the relevant code distances. We first prove a general boundary-scaling bound in Lemma~\ref{lem:routing}: for any family of directional tile codes generated from a fixed directional word \(\mathfrak D\), the number of routing qubits required for nearest-neighbour syndrome extraction satisfies
\[
  n_r=\mathcal O(\sqrt n).
\]
This follows because routing qubits are introduced only near the boundary of the planar layout, whose size grows with the linear dimension of the code. More concretely, the logical operators have a boundary-to-boundary structure, so the \(X\)- and \(Z\)-type distances are controlled by the two bulk dimensions, \(M\) and \(N\), respectively; see~\Cref{tab:routing_dis_M_N}. The same data show that both the standard and optimised routing overheads are linear functions of \(M\) and \(N\). Thus, for a fixed directional word, the routing overhead is naturally expressed as
\[
  n_r=\mathcal O(d_X+d_Z),
\]
and reduces to \(n_r=\mathcal O(d)\) in balanced families where \(d_X\) and \(d_Z\) are comparable. This distance-linear behaviour is shown explicitly in~\Cref{fig:uniform_routing_vs_distance} for the uniform-distance slice \(d\), where both the standard and optimised routing overheads grow approximately linearly with \(d\), while the optimisation yields substantial constant-factor savings across all directional words. This is analogous to the routing overhead recently proved for colour-code constructions~\cite{yoshida2025lowdepthcolorcode} and in correspondence with circuit-qubit number scaling we observed.

Since the routing overhead comes primarily from boundary routing qubits, we reduce it by removing routing sites that are not load-bearing during a syndrome-extraction round. Concretely, if a routing qubit acts only as a terminal relay for a single data or check qubit \(q\), and is not required between two data-check interactions involving \(q\), we absorb this routing site by shifting the corresponding start or end position of \(q\) to that site; see~\Cref{fig:1} (b),~\Cref{fig:nesen_syndrome_extraction}, and~\Cref{fig:routing_optimisation}. We iterate this local simplification until no further terminal routing sites can be removed. The resulting circuit is accepted only after verifying that the measured stabiliser supports, detectors, and logical observables are unchanged. A detailed description of the optimisation procedure is given in Appendix~\ref{app:routing_overhead}.

In multi-patch architectures, part of the routing overhead may be shared with other auxiliary regions already required for logical operations. Lattice-surgery protocols, for example, typically use padding regions between neighbouring patches as auxiliary degrees of freedom for patch merging and logical Pauli measurements~\cite{yang2025planarfaulttolerantquantumcomputation,HorsmanSurgery,QGPU}. From this perspective, boundary routing qubits need not always be viewed as a completely separate overhead: in suitable layouts, they may also serve as inter-patch padding for logical operations, or as auxiliary degrees of freedom for intra-patch logical operations derived from automorphisms~\cite{breuckmann2025logicaloperatorsderivedautomorphisms}; see Appendix~\ref{app:automorphism}. Since these routing qubits can be periodically reset, the same layout is also compatible with leakage-mitigation strategies used in time-dynamic and iSWAP-based circuits~\cite{McEwen2023relaxinghardware,eickbusch_demonstration_2025}.

Routing qubits also provide useful diagnostic information during syndrome extraction. Since they are auxiliary degrees of freedom introduced only to mediate the ordered walk, they can be measured at appropriate time steps. Non-trivial measurement outcomes from routing qubits act as flag information for faults occurring along the routing path and can be supplied directly to the decoder. In our circuit-level simulations, including this routing-qubit information improves decoding performance at relatively high physical error rates, demonstrating that the additional qubits are not merely a geometric overhead but can also provide useful fault-location information; see Appendix~\ref{app:simulation}.

\section{Summary and outlook}

We have introduced directional tile codes: a family of planar qLDPC memories whose stabilisers and syndrome-extraction circuits are specified by the same directional word. Each word defines a bounded-size connected string that serves both as the support of a CSS stabiliser and as the ordered nearest-neighbour walk used to measure it. By exploiting the exchange character of the \(\mathrm{iSWAP}\) gate, these walks generate the required check-data connectivity dynamically in spacetime, rather than requiring a static long-range coupling graph. For a directional word of weight \(w\), one syndrome-extraction round has depth \(w+2\), including check-qubit preparation and measurement, independent of the code size.

This construction weakens the usual dichotomy between surface-code locality and qLDPC efficiency. The surface code is naturally compatible with planar nearest-neighbour superconducting hardware, but pays a large price in encoding rate. Conversely, many high-performance qLDPC constructions achieve much better finite-size parameters, but rely on long-range connectivity or qubit shuttling. Directional tile codes show that this trade-off is not absolute in the finite-size regime relevant to early fault-tolerant devices. We find compact open-boundary instances with high code efficiency. After compiling syndrome extraction into strictly planar nearest-neighbour circuits, this advantage persists at the implemented-layout level: in circuit-level simulations, the best directional tile-code layouts reduce the per-logical per-round logical error rate by around two orders of magnitude relative to rotated surface-code memories at a comparable footprint of about \(40\) circuit qubits per logical qubit on a square grid. Moreover, even when the directional tile codes are restricted to the lower-degree connectivity of a hexagonal grid, they retain nearly a two-order-of-magnitude reduction in logical error rate relative to rotated surface-code memories implemented on the square grid.

Although the implemented layouts already show substantial advantages, further gains are likely possible through more aggressive routing optimisations and a deeper co-design of directional words, planar layouts, syndrome-extraction schedules, and decoders. Several limitations and open directions remain. First, the present numerical evidence is finite-size and based on particular families of directional words; a more systematic search over words, layouts, schedules, and routing optimisations may reveal substantially better instances. Since multiple directional words are allowed at a fixed stabiliser weight, and since these words can realise codes with the same or different logical dimensions, it would also be interesting to develop code-switching techniques between directional tile codes defined by different words. Such techniques could provide additional flexibility for logical operations, layout deformation, or architecture-level compilation. Second, our circuit-level simulations use a simplified noise model. It will be important to optimise these circuits under biased noise, coherent circuit noise, and hardware-calibrated error models. Leakage is another natural direction~\cite{McEwen2023relaxinghardware,eickbusch_demonstration_2025}: although the resettable routing qubits provide possible locations for leakage removal and flag information, leakage has not been explicitly modelled here. Third, these memories should be integrated with logical operations. Promising routes include lattice surgery~\cite{yang2025planarfaulttolerantquantumcomputation} between directional tile patches and intra-patch operations derived from code automorphisms~\cite{breuckmann2025logicaloperatorsderivedautomorphisms,sayginel2025automorphisms}, but their full circuit-level cost in planar layouts remains to be assessed. Finally, architecture-level compilation will be needed to understand how directional tile memories perform inside larger fault-tolerant layouts, where routing regions may also serve as inter-patch workspace.

Taken together, these results suggest that qLDPC-like finite-size advantages can survive the stringent locality constraints of planar superconducting hardware when the code, native gates, and syndrome-extraction dynamics are co-designed. Rather than treating connectivity as a fixed hardware resource, directional tile codes use local exchange gates to generate the required connectivity during the error-correction cycle itself. This provides a route toward lower-overhead fault-tolerant quantum memories on strictly planar nearest-neighbour devices.

\section*{Acknowledgements}
We acknowledge helpful discussions with David Byfield, Stefan Filipp, Julio Carlos Magdalena De La Fuente, György P. Gehér, and Archibald Ruban.
The Berlin team has been supported by the BMFTR (QSolid, MUNIQC-Atoms, PasQuops), the DFG (CRC 183, BoLaCo, and SPP 2514), the Quantum Flagship (Millenion, PasQuanS2), the Munich Quantum Valley, Berlin Quantum, QuantERA (SDPCode), and the European Research Council (DebuQC). TN is supported by the 
Engineering and Physical Sciences Research Council (grant number EP/W524384/1), the University of Edinburgh and Quantinuum. JR is funded by an EPSRC Quantum Career Acceleration Fellowship (grant code: UKRI1224). JR, SK, and BG were supported by the Innovate UK project \enquote{QEC Readout Testbed} [reference number 10151107]. SK and JR were additionally supported by an EPSRC IAA Cross Institutional Project between University of Edinburgh and University of Glasgow.

We thank Georgia M. Nixon, Campbell K. McLauchlan, and Charles Van Rest for coordinating the submission of their independently developed related manuscript, \emph{Vine Codes: Low-Overhead Quantum LDPC Codes on a Planar Square Grid}~\cite{Nixon2026VineCodes}, which addresses closely related questions concerning low-overhead quantum LDPC codes on planar square-grid architectures.

\bibliographystyle{apsrev4-2}
\bibliography{references}
\clearpage
\appendix

\section{Toric directional codes as bivariate-bicycle codes}
\label{app:algebraic_rep}

Toric directional codes~\cite{GeherByfieldRuban2025,rowshan2026structuralanalysisdirectionalqldpc} are obtained by tessellating a generalised torus with directional tiles. This yields a translationally invariant CSS code whose \(X\)-type stabilisers are translates of the directional word  \(\mathfrak D\), while the \(Z\)-type stabilisers are translates of the same string on the dual lattice. This translational invariance reveals their algebraic nature: toric directional codes form a distinguished subfamily of bivariate-bicycle codes~\cite{bravyi2024high,liang2025generalized}.

Algebraically, an \(X\)-type stabiliser is specified by two finite bivariate polynomials \(f(x,y),g(x,y)\in \FF_2[x,y]/(x^l+1,y^m+1)\), which encode its support on the horizontal and vertical edges, respectively, of an \(l\times m\) grid with periodic boundary conditions. The corresponding \(Z\)-type stabiliser is supported on the dual string and is encoded, on horizontal and vertical edges respectively, by the reciprocal polynomials \(\overline{g(x,y)}\) and \(\overline{f(x,y)}\), where for \(f(x,y)=\sum_{i,j} f_{i,j}x^i y^j\) we define \(\overline{f(x,y)} =\sum_{i,j} f_{i,j}x^{-i}y^{-j} =\sum_{i,j} f_{i,j}x^{l-i}y^{m-j}\), with exponents understood modulo \(l\) and \(m\). The CSS parity-check matrices can therefore be written as
\begin{align}
    H_X &= \bigl[\BB(f(x,y)) \mid \BB(g(x,y))\bigr], \\
    H_Z &= \bigl[\BB(\overline{g(x,y)}) \mid \BB(\overline{f(x,y)})\bigr],
\end{align}
where \(\BB\) denotes the binary representation of the quotient ring, defined by \(\BB(x)=S_l\otimes I_m\tand     \BB(y)=I_l\otimes S_m\), with \(S_l\) and $I_l$ the cyclic-shift permutation matrix and the identity matrix of order \(l\), respectively~\cite{bravyi2024high,eberhardt2024logicaloperatorsfoldtransversalgates}. A consequence of using reciprocal polynomials for the \(Z\)-type parity checks is that the corresponding \(Z\)-type stabilisers generally have the same geometric shape as the \(X\)-type stabilisers, up to a \(180^\circ\) rotation. Therefore, if the directional word traces an open, non-self-crossing string, the word must be palindromic for the \(X\)- and \(Z\)-type stabilisers to have identical geometry.

\begin{lemma}[\((w,w)\)-regularity]
    \label{lem:check_weight}
    Consider a toric directional code defined by a directional word \(\mathfrak D\) of weight \(w\). Then every stabiliser generator, corresponding to a row of \(H_X\) or \(H_Z\), has weight \(w\). Moreover, every data qubit participates in exactly \(w\) stabiliser checks in total. Equivalently, the Tanner graph of the code, including both \(X\)- and \(Z\)-type checks, is \((w,w)\)-regular.
\end{lemma}
\begin{proof}
    This follows directly from the bivariate-bicycle representation. Let \(w_f=|\supp\ f|\) and \(w_g=|\supp\ g|\), so that \(w=w_f+w_g\). By construction, \(\BB(f)\) and \(\BB(\overline f)\) are \((w_f,w_f)\)-regular, while \(\BB(g)\) and \(\BB(\overline g)\) are \((w_g,w_g)\)-regular. Hence each row of \(H_X\) and \(H_Z\) has weight \(w_f+w_g=w\). Similarly, for each data qubit, the total number of stabiliser checks it participates in is the sum of the corresponding column weights from \(H_X\) and \(H_Z\), which is again \(w_f+w_g=w\).
\end{proof}

\section{Directional tile codes for biased noise}

For a fixed directional word \(\mathfrak D\), varying the aspect ratio of the rectangular base grid changes the balance between the \(X\)- and \(Z\)-type distances. Below, we list several directional tile codes with strongly unbalanced distances, making them potentially useful for biased-noise architectures~\cite{Putterman, Hajr_2024, Lescanne_2020, XZZX, Tuckett_2018, PhysRevLett.133.110601,roffe_bias-tailored_2023, Xu_2023, das2026clifforddeformedzerorateldpccodes}.

\begin{table}[H]
\centering

{\setlength{\aboverulesep}{0pt}
\setlength{\belowrulesep}{0pt}
\renewcommand{\arraystretch}{1.4}
\begin{tabular}{c| c c |c|c| c}
\toprule
$\mathfrak D,\ w$ &$M$ & $N$ & $\code{n,k,d_X,d_Z}$ & $n_r$ &$(kd_Xd_Z)/n$ \\
\midrule
\midrule

$\texttt{NESEN},\ 5$& 22 & 3 & $\code{259,2,\leq37,4}$ & $58$ & 1.14 \\

\cmidrule{1-6}
$\texttt{N}^2\texttt{ESE}\texttt{N}^2,\ 7$& 22 & 3 & $\code{269,4,\leq 28,5}$ & $87$ & 2.08\\

\cmidrule{1-6}
{$\texttt{N}^2\texttt{E}^2\texttt{S}\texttt{E}^2\texttt{N}^2,\ 9$}& 19 & 4 & $\code{365,10,\leq 31,4}$& $196$ & 3.40\\

\cmidrule{1-6}
{$\texttt{N}^2\texttt{E}^2\texttt{SES}\texttt{E}^2\texttt{N}^2,\ 11$}& 12 & 13 & $\code{537,14,\leq 35,12}$& $216$ & 10.94\\


\bottomrule
\end{tabular}}
\caption{\textbf{Parameters of representative directional tile codes with unbalanced distances.}
Here \(n_r\) denotes the optimised number of routing qubits required for nearest-neighbour syndrome extraction. For each listed directional-word weight \(w\), we perform a brute-force search over all allowed directional words \(\mathfrak D\) and base-grid sizes \(M\times N\) with \(M+N\leq 25\). The table reports the instances with the largest unbalanced code-efficiency ratio \(k d_X d_Z/n\) within this search space, compared with rectangular rotated surface-code patches. For distances above \(20\), the integer-programming distance computation does not terminate within a reasonable time; in these cases we use the probabilistic algorithm \texttt{QDistRnd}~\cite{Pryadko_2022} to obtain distance upper bounds.}

\label{tab:code_para_bias}
\end{table}

\section{Routing overhead of directional tile codes}
\label{app:routing_overhead}

\begin{figure*}
    \centering
    \includegraphics[width=\linewidth]{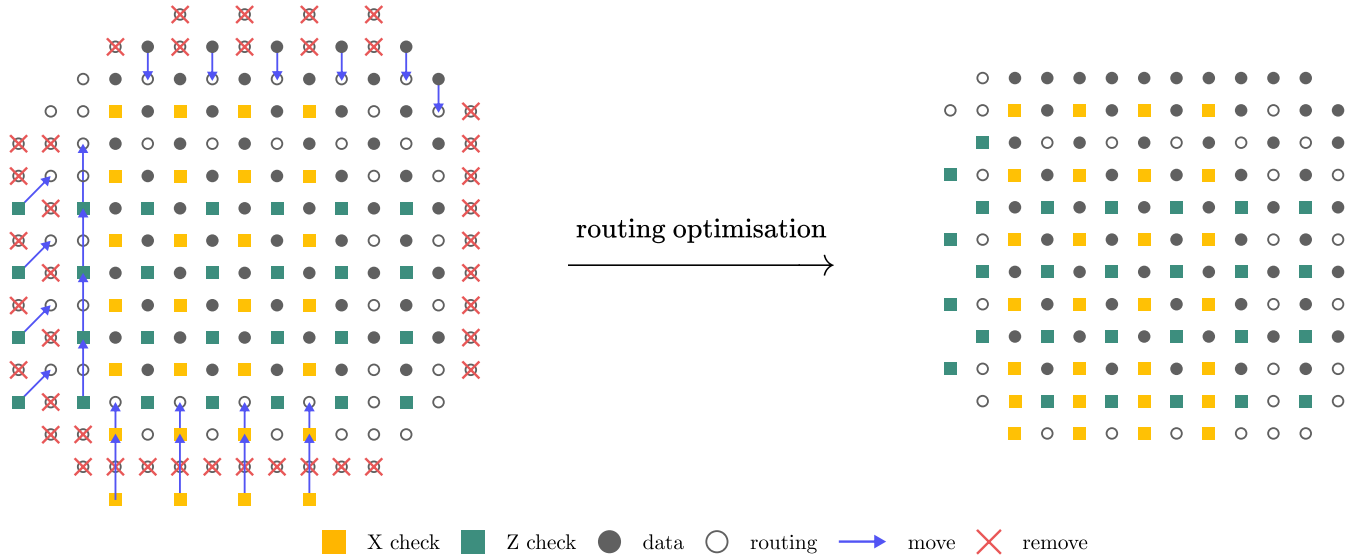}
    \caption{\textbf{Routing optimisation.}
    Example of the boundary-routing optimisation used to reduce the number of auxiliary qubits in the implemented layout and circuit for the \(\texttt{N}^2\texttt{ESE}\texttt{N}^2\) directional tile code with parameters \(\code{60,4,5}\) as presented in \Cref{fig:1}. The left panel shows the standard layout before optimisation: empty circles denote routing sites, red crosses mark routing sites that are removed, and blue arrows indicate shifts of data or check start positions. The right panel shows the optimised layout after pruning terminal routing sites that are not load-bearing during the syndrome-extraction walk. The optimisation preserves the measured stabiliser supports, detectors and logical observables, while reducing the total number of circuit qubits.}
    \label{fig:routing_optimisation}
\end{figure*}

In this section, we bound the routing overhead of directional tile codes and describe the boundary optimisation used in the implemented circuits. The key point is that routing qubits are required only near the boundary of the planar layout. Consequently, for a fixed directional word, the routing overhead grows with the linear size of the patch rather than with its area.

\begin{lemma}[$\sqrt n$-scaling of the routing overhead]
    \label{lem:routing}
    Consider a family of directional tile codes obtained from a fixed directional word \(\mathfrak D\) with fixed tile size \(B\times B\). Let the bulk patch have size \(M\times N\), with bounded aspect ratio, and let \(n\) denote the number of data qubits. Then the number \(n_r\) of routing qubits required for nearest-neighbour syndrome extraction satisfies
    \[
        n_r=\mathcal O(\sqrt n).
    \]
\end{lemma}

\begin{proof}
    Let \(L_h=M+2(B-1)\) and \(L_v=N+2(B-1)\). Since \(B\) is fixed, \(L_h=M+\mathcal O(1)\) and \(L_v=N+\mathcal O(1)\). The total number of hardware locations on vertices, edges and plaquettes is
    \eqa{
        n_{\rm hw}
        &=L_hL_v+(2L_hL_v-L_h-L_v)+(L_h-1)(L_v-1)\\
        &=4MN+\mathcal O(M+N).
    }
    The bulk contains \(MN\) stabiliser anchors of each type. Thus there are \(2MN\) bulk check qubits, and since the corresponding stabilisers are linearly independent,
    \[
        n\ge 2MN=\Theta(MN).
    \]
    Hence, data and check qubits together occupy at least \(4MN\) hardware locations. Therefore
    \[
        n_r\le n_{\rm hw}-4MN=\mathcal O(M+N).
    \]
    Under bounded aspect ratio, \(M+N=\mathcal O(\sqrt{MN})\), and hence
    \[
        n_r=\mathcal O(\sqrt n).
    \]
\end{proof}

The boundary-scaling bound is naturally expressed in terms of distance. The logical operators of these codes have a boundary-to-boundary structure, so the \(X\)- and \(Z\)-type distances are controlled primarily by the two bulk dimensions \(M\) and \(N\), respectively. This is reflected in the empirical fits in~\Cref{tab:routing_dis_M_N}. The same table shows that both the standard routing overhead \(\widetilde n_r\) and the optimised routing overhead \(n_r\) are affine functions of \(M\) and \(N\). Thus, for a fixed directional word, the routing overhead scales linearly with the relevant boundary-to-boundary distance scale. In particular, for the balanced families used in our main overhead comparisons, this gives
\[
    n_r=\mathcal O(d).
\]
This behaviour is shown explicitly in~\Cref{fig:uniform_routing_vs_distance} for uniform-distance instances, \(d=d_X=d_Z\), where both the standard and optimised routing overheads grow approximately linearly with \(d\). The optimisation gives substantial constant-factor savings across all directional words.

\begin{table*}[]
\centering
{\setlength{\aboverulesep}{0pt}
\setlength{\belowrulesep}{0pt}
\renewcommand{\arraystretch}{1.5}
\begin{tabular}{c|c|c|c|c}
\toprule
\(\mathfrak D,\ w\) & \(\widetilde{n_r}\) & \(n_r\) & \(d_X\) & \(d_Z\)\\
\midrule
\midrule

\(\texttt{N}^2\texttt{ESE}\texttt{N}^2,\ 7\) 
& \(8M+8N+36\) 
& \(3M+2N+15\) 
& \(\simeq 1.10M-0.08N+2.25\) 
& \(\simeq 1.01N-0.01M+2.07\) \\

\cmidrule{1-5}
\(\texttt{N}^2\texttt{E}^2\texttt{S}\texttt{E}^2\texttt{N}^2,\ 9\) 
& \(18M+8N+72\) 
& \(7M+6N+39\) 
& \(\simeq 1.51M-0.23N+2.54\) 
& \(\simeq 0.57N-0.05M+1.68\) \\

\cmidrule{1-5}
\(\texttt{N}^2\texttt{E}^2\texttt{SES}\texttt{E}^2\texttt{N}^2,\ 11\) 
& \(24M+8N+98\) 
& \(9M+4N+56\) 
& \(\simeq 2.30M-0.12N+3.77\) 
& \(\simeq 0.87N-0.03M+0.74\) \\
\bottomrule
\end{tabular}}

\caption{\textbf{Routing overhead and distance scaling with bulk dimensions.}
For each directional word \(\mathfrak D\), the table gives the routing overhead before and after optimisation as functions of the bulk dimensions \(M\) and \(N\). Here \(\widetilde n_r\) denotes the routing overhead of the standard layout, while \(n_r\) denotes the optimised routing overhead after removing unnecessary routing qubits. The last two columns give empirical linear fits for the \(X\)- and \(Z\)-type distances, showing that \(d_X\) is primarily controlled by \(M\) and \(d_Z\) by \(N\).}
\label{tab:routing_dis_M_N}
\end{table*}

We reduce the routing overhead by generating and testing simplified routing layouts. We call the primary  simplification  ``route-window shortening''. For each boundary stabiliser, we identify the first and last steps of the directional word at which it interacts with surviving data qubits. Any prefix or suffix outside this active window is not needed to measure the stabiliser support, so the check start position can be shifted and the corresponding terminal routing sites can be removed, provided that no check-start collision or data-check overlap is introduced. We then apply the same principle to terminal routing sites used only by a single data or check qubit \(q\): if the site is not required between two data-check interactions involving \(q\), it is absorbed into the trajectory of \(q\) by moving the appropriate start position to that site. For data qubits, this is implemented as a data-start shift onto a former routing coordinate.

After each proposed removal, we try to greedily remove more routing qubits, by iteratively deleting routing qubits near the edges of the layout and shifting the remaining qubits to fill the newly introduced hole. We then run the directional syndrome extraction normally for two rounds (forwards and backwards) and accept the move only if every stabiliser still interacts with exactly the data qubits specified by the original parity-check matrix. We call this method ``trace pruning'' as it relies on the trace of the directional word on the grid fixing future steps. The intuition for this is that even though in the original layout some routing qubits might appear necessary, through the displacement of early directional word steps, their neighbour routing qubits can fill the gap instead when near the middle of the schedule. Several candidate circuits are generated in this way, including variants with and without data-start shifts and with different combinations of greedy deletions. The final candidate is accepted only if the measured stabiliser supports, detectors, and logical observables are unchanged, and if the detector error model remains deterministic. Among all valid candidates, we choose the one with the smallest number of physical qubits. We show the unoptimised layout, together with the routing-site moves and deletions produced by the combined optimisation procedure, in Figure~\ref{fig:routing_optimisation}.

\begin{example}[Routing optimisation of the $\code{323,14,15}$ directional tile code]
For the $\code{323,14,15}$ directional tile code generated by the weight-$11$ directional word \(\texttt{N}^2\texttt{E}^2\texttt{SES}\texttt{E}^2\texttt{N}^2\), the unoptimised layout contains \(346\) routing sites. By combining the two optimisation procedures and repeating them with different update orders, the routing overhead can be reduced iteratively. Applying route-window shortening first reduces the number of routing sites from \(346\) to \(235\). Subsequent trace pruning of the remaining routing sites further reduces this number to \(165\). Thus, the combined optimisation removes \(181\) routing sites in total, corresponding to a reduction of approximately \(52\%\) relative to the unoptimised layout.
\end{example}

\begin{figure}
    \centering
    \includegraphics[width=\linewidth]{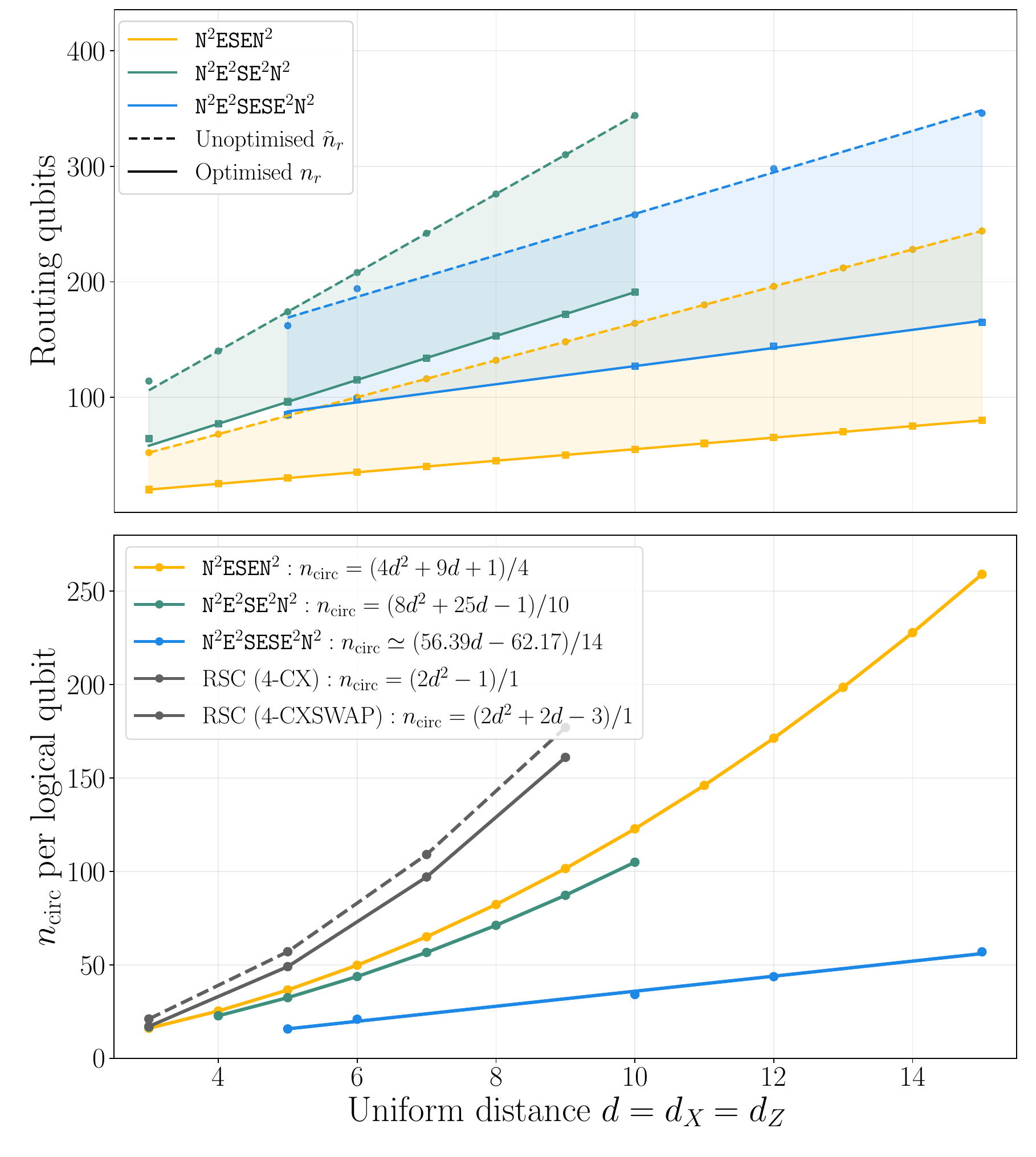}
    \caption{\textbf{Routing and circuit-qubit overheads as a function of distance.}
    \texttt{Top:} Routing-qubit overhead for representative directional tile-code families with uniform distance \(d\). Dashed curves show the standard-layout overhead \(\widetilde n_r\), while solid curves show the optimised overhead \(n_r\) after boundary-routing simplification. Both scale approximately linearly with \(d\), with the optimisation yielding substantial savings across all directional words.
    \texttt{Bottom:} Total circuit-qubit number \(n_{\mathrm{circ}}\) for the corresponding optimised nearest-neighbour layouts. The \(w=7\) data follow the quadratic fit exactly over the shown range, and the \(w=9\) data apart from the \(d=3\) instance, whereas the \(w=11\) family is well described by approximately linear fits over the finite-distance regime considered.}
    \label{fig:uniform_routing_vs_distance}
\end{figure}

\section{Logical gates from automorphisms and derived automorphisms}
\label{app:automorphism}

We briefly discuss possible mechanisms for implementing logical operations in the directional tile-code families considered in this work. Since these codes are constructed as tile codes, it is natural to use the structural description of logical operators developed in Refs.~\cite{steffan2025high,breuckmann2025logicaloperatorsderivedautomorphisms}. In particular, Ref.~\cite{breuckmann2025logicaloperatorsderivedautomorphisms} gives a canonical symplectic basis of logical Pauli operators for tile codes, with representatives supported near the lattice boundaries and generated by a finite cellular-automaton rule. We use this basis to describe the logical Pauli operators of our directional tile codes and to visualise representative logical operators on the physical qubit lattice; examples are shown in~\Cref{fig:canon_logicals}.

\begin{figure*}
    \centering
    \includegraphics[width=\linewidth]{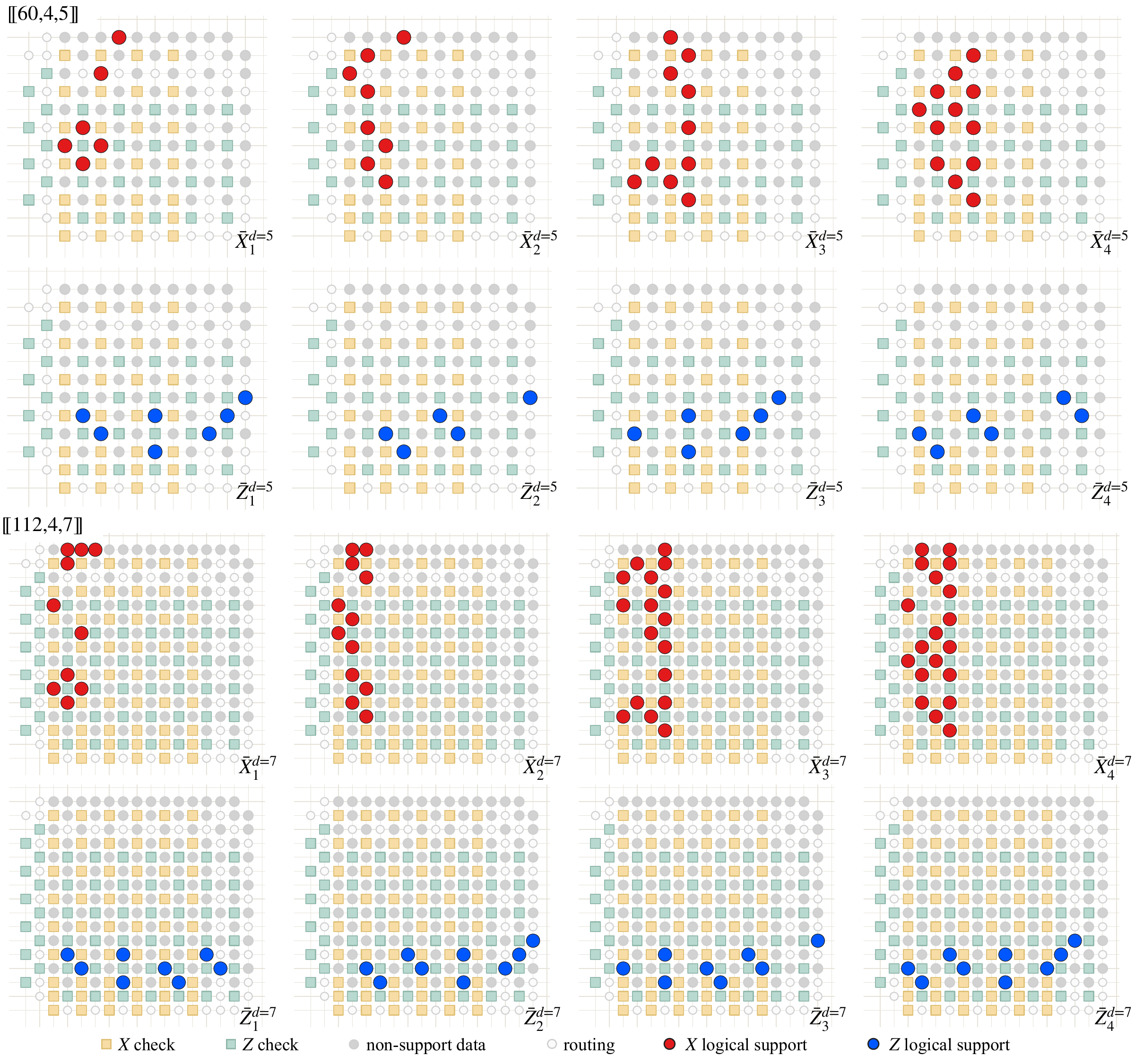}
    \caption{\textbf{Canonical logical Pauli bases for the $\texttt{N}^2\texttt{ESE}\texttt{N}^2$ directional tile codes with stabiliser weight $w=7$.}
    We show two representative codes from this directional-word family, with parameters $\code{60,4,5}$ and $\code{112,4,7}$. Supports of $X$-type logical operators are shown in red, while supports of $Z$-type logical operators are shown in blue. The displayed operators form canonical logical Pauli bases constructed using the method of Ref.~\cite{breuckmann2025logicaloperatorsderivedautomorphisms}, such that $\bar X_i$ and $\bar Z_j$ anti-commute if and only if \(i=j\), and otherwise commute.}
    \label{fig:canon_logicals}
\end{figure*}

A first route to logical operations is provided by the \textit{derived-automorphism}~\cite{breuckmann2025logicaloperatorsderivedautomorphisms}. In this approach, the tile-code patch is temporarily enlarged along one boundary, implementing a controlled deformation of the code. This deformation induces a non-trivial Clifford action on the encoded qubits. When expressed in the canonical logical basis of Ref.~\cite{breuckmann2025logicaloperatorsderivedautomorphisms}, the resulting action decomposes into products of logical $\mathrm{CNOT}$ gates, while remaining compatible with a low-overhead fault-tolerant implementation.

Some of the directional words studied here have an additional geometric symmetry: they are palindromic about their central step. For a directional tile code with directional word $\mathfrak D=\vec d_1\vec d_2\cdots \vec d_w$, the sequence satisfies
\[
    \vec d_i=\vec d_{w+1-i},
\]
possibly after identifying directions under the corresponding reflection of the tile. This symmetry induces a reflection of the physical qubit layout,
\[
    \rho:q_i\mapsto q_{w+1-i},
\]
on the tile-code lattice. Since the reflected directional word defines the same local stabiliser constraints, the corresponding physical qubit permutation preserves the stabiliser group and therefore defines a \textit{code automorphism}~\cite{sayginel2025automorphisms}.

This automorphism can in principle be realised as a permutation of physical qubits, or, equivalently, by relabelling physical qubits. Although the reflected qubits generally lie on opposite sides of the directional word, so the permutation is not geometrically nearest-neighbour in the original planar layout, we do not need to realise it as an actual SWAP circuit. Instead, each qubit is re-indexed by its reflected position, and the stabiliser generators, decoder information, and Pauli frame are updated accordingly. Using the canonical logical Pauli basis, we find that this reflection automorphism acts non-trivially on the encoded qubits. The induced logical Clifford operation decomposes into a product of logical SWAP and CNOT gates.

\begin{table*}[]
    \centering
    \setlength{\aboverulesep}{0pt}
    \setlength{\belowrulesep}{0pt}
    \renewcommand{\arraystretch}{1.4}
    \begin{tabular}{c|| c c c c c c c}
    \toprule
    \(\code{n,k,d}\) & $\code{24,4,3}$ & $\code{40,4,4}$ & $\code{60,4,5}$ & $\code{84,4,6}$ & $\code{112,4,7}$ & $\code{144,4,8}$ & $\code{180,4,9}$ \\

    \midrule
    \(|\mathrm{Aut}|\) & \(2^{9}\) & \(2^{13}\) & \(2^{17}\) & \(2^{21}\) & \(2^{24}\) & \(2^{23}\) & \(2^{33}\) \\
    \bottomrule
    \end{tabular}
    \caption{\textbf{Automorphism-group orders for selected $\texttt{N}^2\texttt{ESE}\texttt{N}^2$ directional tile codes.}
    For each code instance \(\code{n,k,d}\), the table gives the order of the physical-qubit permutation + single qubit Clifford gate group preserving the stabiliser code.}
    \label{table:w7_aut_orders}
\end{table*}

We complement these analytic observations with an automated search for further code automorphisms. Following the framework of Ref.~\cite{sayginel2025automorphisms}, one can represent a stabiliser code as a binary linear code, compute generators of its automorphism group, impose the allowed Clifford constraints, and determine the induced logical action of each physical symmetry. Applying this approach to our code families reveals multiple automorphism generators and large symmetry groups, suggesting that the palindromic reflection is only one instance of a richer algebraic structure. These additional symmetries may provide useful ingredients for fault-tolerant Clifford computation~\cite{he_extractors_2025} or symmetry-assisted decoding~\cite{koutsioumpas_automorphism_2025}. Examples of automorphism-group orders for the $\texttt{N}^2\texttt{ESE}\texttt{N}^2$ directional tile code family with stabiliser weight $w=7$ are listed in~\Cref{table:w7_aut_orders}.

\section{Circuit-level memory simulations with only nearest-neighbour operations}
\label{app:simulation}

\begin{table}[]
\vspace{1em}
    \centering
    \setlength{\aboverulesep}{0pt}
    \setlength{\belowrulesep}{0pt}
    \renewcommand{\arraystretch}{1.3}
    \begin{tabular}{c|c|c|c}
    \toprule
    Code family & No. patches & \(\code{n,k,d}\) & \(p_L(p=0.001)\) \\
    \midrule
    \midrule

    \multirow{4}{*}[-0.2ex]{\shortstack{
        \(\texttt{N}^2\texttt{ESE}\texttt{N}^2,\ 7\)\\
        (square)}}
    & \multirow{4}{*}[-0.2ex]{\(35\)}
    & \(\code{24,4,3}\)  & \(9.00\times10^{-4}\) \\
    & & \(\code{40,4,4}\)  & \(2.87\times10^{-4}\) \\
    & & \(\code{60,4,5}\)  & \(4.54\times10^{-5}\) \\
    & & \(\code{84,4,6}\)  & \(1.36\times10^{-5}\) \\
    \midrule

    \multirow{4}{*}[-0.2ex]{\shortstack{
        \(\texttt{N}^2\texttt{E}^2\texttt{S}\texttt{E}^2\texttt{N}^2,\ 9\)\\
        (square)}}
    & \multirow{4}{*}[-0.2ex]{\(14\)}
    & \(\code{52,10,3}\)   & \(5.30\times10^{-4}\) \\
    & & \(\code{76,10,4}\)   & \(1.65\times10^{-4}\) \\
    & & \(\code{115,10,5}\)  & \(2.17\times10^{-5}\) \\
    & & \(\code{162,10,6}\)  & \(9.10\times10^{-6}\) \\
    \midrule

    \multirow{4}{*}[-0.2ex]{\shortstack{
        \(\texttt{N}^2\texttt{E}^2\texttt{SES}\texttt{E}^2\texttt{N}^2,\ 11\)\\
        (square)}}
    & \multirow{4}{*}[-0.2ex]{\(10\)}
    & \(\code{74,14,5}\)    & \(2.89\times10^{-4}\) \\
    & & \(\code{104,14,6}\)   & \(2.98\times10^{-5}\) \\
    & & \(\code{182,14,10}\)  & \(6.98\times10^{-6}\) \\
    & & \(\code{241,14,12}\)  & \(1.37\times10^{-6}\) \\
    \midrule

    \multirow{4}{*}[-0.2ex]{\shortstack{
        \(\texttt{ENE}\texttt{S}^2\texttt{W}\texttt{S}^2\texttt{ENE},\ 11\)\\
        (hexagon)}}
    & \multirow{4}{*}[-0.2ex]{\(14\)}
    & \(\code{52,10,3}\)   & \(2.71\times10^{-4}\) \\
    & & \(\code{74,10,4}\)   & \(5.45\times10^{-5}\) \\
    & & \(\code{104,10,5}\)  & \(1.56\times10^{-5}\) \\
    & & \(\code{199,10,7}\)  & \(2.30\times10^{-6}\) \\
    \midrule

    \multirow{3}{*}[-0.2ex]{\shortstack{
        Rotated surface codes\\
        (\(4\text{-}\mathrm{CX}\))}}
    & \multirow{3}{*}[-0.2ex]{\(140\)}
    & \(\code{9,1,3}\)            & \(1.24\times10^{-3}\) \\
    & & \(\code{25,1,5}\)           & \(8.72\times10^{-5}\) \\
    & & \(\code{49,1,7}^{\dagger}\) & \(5.44\times10^{-6}\) 
    \\
    \midrule

    \multirow{3}{*}[-0.2ex]{\shortstack{
        Rotated surface codes\\
        (\(4\text{-}\mathrm{CXSWAP}\))}}
    & \multirow{3}{*}[-0.2ex]{\(140\)}
    & \(\code{9,1,3}\)            & \(1.09\times10^{-3}\) \\
    & & \(\code{25,1,5}\)           & \(9.33\times10^{-5}\) \\
    & & \(\code{49,1,7}^{\dagger}\) & \(6.07\times10^{-6}\) 
    \\
    \bottomrule
    \end{tabular}

    \caption{\textbf{All codes included in the circuit-level simulations of \Cref{fig:footprint}.}
    Several patches of each code are used to build a quantum memory of \(140\) logical qubits, which is the least common multiple of the logical dimensions of the codes listed, in order to make a fair comparison. The grid connectivity is indicated below each directional word, and both rotated surface-code implementations use the square grid. Here \(p_L(p=0.001)\) denotes the circuit-level logical error rate averaged over $X$ and $Z$ bases per-logical per-round at physical error rate \(p=0.001\). \(^\dagger\)These codes lie to the right of the plotting range in \Cref{fig:footprint}, and their data points are omitted for clarity.}
    \label{tab:full_sims_codes}
\end{table}

\begin{table}[H]
\centering
    \begin{tabular}{llc}
    \toprule
    Operation & Error channel & Probability \\ \midrule
    \midrule
    Two-qubit Cliffords & two-qubit depolarising & $p$ \\ 
    One-qubit Clifford & one-qubit depolarising & 
    $p$\\
    Reset in $X$ ($Z$) basis & phase (bit) flip & $p$ \\
    Measure in $X$ ($Z$) basis & phase (bit) flip & $p$ \\
    Idling qubits  & one-qubit depolarising & $p$ \\
    \bottomrule
    \end{tabular}
    \caption{\textbf{The uniform noise model.} Operations, their associated error channels and probabilities for the uniform noise model.}
    \label{table:uniform_noise}
\end{table}

For our numerical evaluations, we perform circuit-level memory simulations with the number of syndrome-extraction rounds set equal to the code distance \(d\) for each code instance. The circuits are compiled using CXSWAP gates and simulated under the uniform depolarising noise model specified in Table~\ref{table:uniform_noise}. In our main results, shown in Figure~\ref{fig:footprint}, we fix the physical noise parameter to \(p=0.001\).

All simulations are executed using \texttt{Stim}~\cite{gidney2021stim}. For each circuit, we sample up to 10 million shots, terminating early once 30 logical errors have been observed. For the surface code patches, we use correlated pymatching \cite{HiggottGidney2025}. Decoding for the directional tile codes is performed with \texttt{VibeLSD}~\cite{Kouts_vibedec}, which performs a parallel ensemble of serial min-sum schedules with localised statistics post-processing on unconverged shots~\cite{hillmann_localized_2025}. Here we use an ensemble of size of 16 for 5 iterations, for the unconverged shots we extend the ensemble to 32 and subsequently to 64 and 128 paths progressively. For shots that still do not converge on any path, we fall back to \texttt{LSD} order 0. Although the detector error models for directional tile-code circuits are denser than those of standard surface-code circuits, due to the longer directional walks and the presence of routing-qubit detectors, we find that this generic \texttt{BP}-based decoder already gives strong circuit-level performance. We have not attempted to optimise the decoder for these circuits, and further improvements may be possible with decoding methods tailored to the denser spacetime structure of directional tile codes.
\end{document}